\documentclass{article}
\usepackage[utf8]{inputenc}
\usepackage{graphicx}
\usepackage{comment}
\graphicspath{ {./images/} }
\usepackage{hyperref}
\usepackage{authblk}
\usepackage{csquotes}
\usepackage{listings}

\title{\textsc{3-coloring} in Time $\mathcal{O}^*(1.3217^n)$}
\author[1]{Lucas Meijer}
\affil[ ]{\normalsize \href{mailto:l.meijer2@uu.nl}{l.meijer2@uu.nl}}
\affil[1]{Utrecht University, Department of Information and Computing Sciences}
\date{ \today}
\usepackage[english]{babel}
\usepackage{geometry}
\usepackage{amsmath}
\usepackage{amsthm}
\usepackage[table,xcdraw,dvipsnames]{xcolor}
\usepackage{biblatex}
\usepackage{tikz}
\usepackage{float}
\usepackage[normalem]{ulem}
\useunder{\uline}{\ul}{}
\newtheorem{theorem}{Theorem}[section]

\newtheorem{lemma}[theorem]{Lemma}
\theoremstyle{definition}
\newtheorem{definition}{Definition}[section]
\begin{document}
\maketitle
\begin{abstract}
    We propose a new algorithm for \textsc{3-coloring} that runs in time $\mathcal{O}^*(1.3217^n)$.
    For this algorithm, we make use of the time $\mathcal{O}^*(1.3289^n)$ algorithm for \textsc{3-coloring} by Beigel and Eppstein.
    They described a structure in all graphs, whose vertices could be colored relatively easily.
    In this paper, we improve upon this structure and present new ways to determine how the involved vertices reduce the runtime of the algorithm.
\end{abstract}

\section{Introduction}
\begin{figure}[H]
\centering
\includegraphics[page=1]{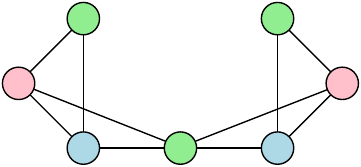}
\hspace{1cm}%
\includegraphics[page=2]{Images/3color.pdf}
\caption{Left: a 3-colorable graph. Right: a non-3-colorable graph.}
\label{fig:excolor}
\end{figure}
The \textsc{3-coloring} problem is one of the most fundamental problems in graph theory.
The \textsc{3-coloring} problem poses the following question:
given a graph $G$ with $n$ vertices, can we assign one of three colors (i.e. red, green, or blue) to every vertex, such that no two adjacent vertices are assigned the same color.

Notably, \textsc{3-coloring} is a special case of the  \textsc{graph coloring} problem.
In \textsc{graph coloring}, we aim to minimize the number of colors we need to color a graph, such that no two adjacent vertices receive the same color.
The \textsc{graph coloring} is a classic problem in complexity theory; it is one of Karp's original 21 \textsc{NP}-complete problems \cite{Karp1972}.
A year after Karp, Lovász showed that \textsc{3-coloring} is \textsc{NP}-complete \cite{Lovasz1973}.
Specifically, it is the lowest number of colors for which it is \textsc{NP}-complete whether a graph can be colored using this number of colors:
\textsc{0-coloring} and \textsc{1-coloring} are both trivially decided in polynomial time, while a \textsc{2-coloring} graphs can be found, if one exists, by greedily coloring vertices.

\subsection{History of \textsc{3-coloring}}
\begin{table}[]
\centering
\begin{tabular}{lll}
\textbf{Year} & \textbf{Author}     & \textbf{Time}                                    \\ \hline
1976          & Lawler              & $\mathcal{O}^*(1.4425^n)$ \\
1994          & Schiermeyer         & $\mathcal{O}^*(1.415^n)$   \\
2000          & Beigel and Eppstein & $\mathcal{O}^*(1.3289^n)$
\end{tabular}
\caption{\label{tab:overview}Overview of previous algorithms for \textsc{3-coloring}.}
\end{table}
Trivially, \textsc{3-coloring} can be solved in time $\mathcal{O}^*(3^n)$, by attempting every assignment of colors, and validating in polynomial time whether it leads to a correct solution.
The first non-trivial algorithm was created in 1976 by Lawler \cite{Lawler1976}.
He figured that \textsc{3-coloring} can be solved by iterating over all maximal independent sets with polynomial delay.
For each maximal independent set $I$ in a graph $G$, we assign one color to all vertices in the maximal independent set.
Then, we verify whether $G \backslash I$ is 2-colorable.
If there does not exist an independent set $I$, such that $G \backslash I$ is 2-colorable, $G$ is not 3-colorable.
Moon and Moser showed there exist at most $\mathcal{O}(3^{n/3}) = \mathcal{O}(1.44225^n)$ maximal independent sets in any graph \cite{Moon1965}, which we can iterate over in at most time $\mathcal{O}(3^{n/3}) = \mathcal{O}(1.44225^n)$ \cite{Johnson1988}.
Verifying whether $G \backslash I$ is 2-colorable takes polynomial time, so Lawler's algorithm runs in time $\mathcal{O}^*(3^{n/3})$.

In 1994, Schiermeyer found an improved algorithm to solve \textsc{3-coloring} in time $\mathcal{O}^*(1.415^n)$ \cite{Schiermeyer1994}.
Critically, Schiermeyer's algorithm was based on the idea that for any vertex, its set of neighbors must be 2-colorable if the graph is 3-colorable.

Before the time $\mathcal{O}^*(1.3217^n)$ algorithm presented in this paper, the best known algorithm for \textsc{3-coloring} was Beigel and Eppstein's time $\mathcal{O}^*(1.3289^n)$ algorithm \cite{Beigel2005}.
Beigel and Eppstein published this paper in 2000.
They created an algorithm that solves the \textsc{(3,2)-Constraint Satisfaction Problem} (\textsc{(3,2)-CSP}) in time $\mathcal{O}^*(1.36444^n)$.
The \textsc{3-coloring} problem can be reduced to \textsc{(3,2)-CSP} without increasing the instance size, implying that \textsc{3-coloring} can be solved in time $\mathcal{O}^*(1.36444^n)$ as well.
Furthermore, Beigel and Eppstein found that there must exist vertices that are easier to color than by using the \textsc{(3,2)-CSP} algorithm.
By first coloring a subset of the colors, they reach the time $\mathcal{O}^*(1.3289^n)$ algorithm.
Our algorithm will expand upon Beigel and Eppstein's algorithm and further improve it.

\subsection{State of Graph Coloring}
Given a graph $G$, we can determine whether $G$ can be colored using at most $k$ colors.
Notably, in \textsc{graph coloring}, we want to minimize the value of $k$, such that we can color $G$ using $k$ colors.
The \textsc{graph coloring} problem can be solved in time $\mathcal{O}^*(2^n)$ by using the technique of inclusion-exclusion \cite{Bjrklund2009}.
Surprisingly, for any $k \geq 7$, the time $\mathcal{O}^*(2^n)$ algorithm is the fastest known.

The fastest known algorithm for \textsc{4-coloring} runs in time $\mathcal{O}^*(1.7272^n)$ and was published by Fomin et al. in 2007 \cite{Fomin2007}.
Fomin et al. showed that a graph will either have a low number of maximal independent sets or a low pathwidth.
They created an algorithm that iterates over all maximal independent sets and one that is fixed-parameter tractable in the pathwidth of the graph;
depending on the graph, they choose to run either of the two.
The algorithm that iterates over all independent sets uses the best known algorithm for \textsc{3-coloring} as a subroutine.
As such, our improvement to \textsc{3-coloring} also improves \textsc{4-coloring}.
We can now solve \textsc{4-coloring} in time $\mathcal{O}^*(1.7247^n)$.

Finally, for \textsc{5-coloring} and \textsc{6-coloring}, recent developments by Zamir showed that both \textsc{5-coloring} and \textsc{6-coloring} can be solved in time $\mathcal{O}^*((2-\epsilon)^n)$ for a small value $\epsilon$ (both with a different $\epsilon$) \cite{https://doi.org/10.48550/arxiv.2007.10790}.
For any $k \geq 7$, it is still an open question whether there exists a time $\mathcal{O}^*((2-\epsilon)^n)$ algorithm to solve $k$-\textsc{coloring}.

\subsection{Our Contribution}
In this paper, we prove the following theorem:
\begin{theorem}
\label{the:1}
There is an algorithm for \textsc{3-coloring} running in time $\mathcal{O}^*(1.3217^n)$ on $n$-vertex graphs.
\end{theorem}

We improve upon Beigel and Eppstein's time $\mathcal{O}^*(1.3289^n)$ algorithm.
Their algorithm has been the best known result for over twenty years.
As such, we believe our improvement in the runtime is an important result for the \textsc{3-coloring} problem.

\subsection{Organization}
We expand upon many of Beigel and Eppstein ideas; we discuss their algorithm in Section \ref{sec:2}.
We introduce a new graph structure, the maximal low-magnitude bushy forest, in Section \ref{sec:forest}, which helps us find vertices that can be colored relatively easily.
We determine how the maximal low-magnitude bushy forest allows us to color vertices more quickly in Section \ref{sec:analysis}.
We combine our findings in Section \ref{sec:lp}, where we analyze the runtime of our algorithm by creating a linear program.
We summarize our algorithm in \ref{sec:conclusion}.

\section{Summary of Beigel and Eppstein's algorithm}
\label{sec:2}
In this section, we discuss the time $\mathcal{O}^*(1.3289^n)$ algorithm to solve \textsc{3-coloring} by Beigel and Eppstein \cite{Beigel2005}. 
Before the improvements presented in this paper, the algorithm from Beigel and Eppstein was the fastest known algorithm to solve \textsc{3-coloring}. 
Notably, we created our algorithm by improving upon their algorithm. 
We will present the concepts from Beigel and Eppstein that are important to our improvements. 
Some lemmas are slightly modified from Beigel and Eppstein's algorithm, to optimize them for our algorithm.

Beigel and Eppstein showed that the \textsc{(3,2)-Constraint Satisfaction Problem} (\textsc{(3,2)-CSP}) can be used as a black box to solve \textsc{3-coloring} efficiently.
The \textsc{(3,2)-Constraint Satisfaction Problem} consists of a set of variables, each of which must be assigned one of at most three colors.
The combination of a variable and one of its colors is called a variable-color pair.
The \textsc{(3,2)-Constraint Satisfaction Problem} contains constraints between two variable-color pairs: 
not both vertices can be assigned the color in their respective variable-color pairs. 
We transform \textsc{3-coloring} into \textsc{(3,2)-Constraint Satisfaction Problem} as follows:
\begin{enumerate}
    \item Every variable represents a vertex.
    \item All variables have the same domain of colors as the vertex they represent.
    \item For any two variables representing adjacent vertices, we add a constraint that they cannot both be the same color.
\end{enumerate}

\begin{theorem}[{\cite[Theorem~1]{Beigel2005}}]
There is an algorithm for the \textsc{(3,2)-Constraint Satisfaction Problem} running in time $\mathcal{O}^*(1.36443^n)$, where $n$ is the number of variables.
\end{theorem}

Beigel and Eppstein showed that the \textsc{(3,2)-Constraint Satisfaction Problem} can be solved in time $\mathcal{O}^*(1.36443^n)$.
The reduction does not increase the instance size, so \textsc{3-coloring} can be solved in time $\mathcal{O}^*(1.36443^n)$ as well.

Notably, given a partially colored graph, we can also use an algorithm for \textsc{(3,2)-CSP} to determine whether there exists an assignment of colors to the uncolored vertices, such that it becomes a valid 3-coloring.
In Beigel and Eppstein's algorithm for the \textsc{(3,2)-Constraint Satisfaction Problem}, there exists a reduction rule that can remove any variable with two or fewer colors from the instance in polynomial time.
Neighbors of colored vertices have at most two possible colors. 
So, these vertices can be removed from the instance.
For instance, consider a small set of vertices $S$ that have a large set of neighbors $T$.
There exist at most $3^{|S|}$ valid color assignments for the vertices in $S$. By iterating over these color assignments, we could solve \textsc{3-coloring} in time $\mathcal{O}^*(1.36443^{|V|-|S|-|T|}\cdot3^{|S|})$.

Clearly, when $|T|$ is sufficiently large compared to $|S|$, this will improve the runtime of the algorithm.
In this section, we will explain how Beigel and Eppstein found a small set of vertices to color to improve the algorithm for \textsc{3-coloring} and our adjustments to their lemmas to improve them within the context of the new algorithm.

\subsection{Work Factor}
Beigel and Eppstein's algorithm use many branching rules to eliminate certain cases from the graph.
A branching rule is a strategy to solve an instance by solving several instances recursively.
To analyze these branches, they use the work factor:
\begin{definition}[Work Factor]
    The \textit{work factor}, $\lambda(r_1, r_2, \dots)$, denotes the complexity of a branching rule.
    We explore some number of branches in a branching rule.
    Each of the branches explores an instance with a reduced instance size.
    In the work factor, each value $r_1, r_2, \dots$ denotes the reduction in instance size in a branch.
    Each work factor $\lambda(r_1, r_2, \dots) = c$, where $c$ is the largest zero of the function $f(x)=1-\sum x^{-r_i}$.
    The work factor assists us in calculating the runtime of an algorithm.
    Every step, some branching rule is applied.
    Each branching rule will have some work factor $\lambda$: the algorithm will run within time $\mathcal{O}^*((\max \lambda)^n)$.
    So, if for all work factors $\lambda \leq 1.3217$, the algorithm runs in time $\mathcal{O}^*(1.3217^n)$.
\end{definition}

\subsection{Removing Low-Degree Vertices}
By definition, low-degree vertices have few neighbors.
So, we want to avoid trying all of their color assignments.
Indeed, assigning a color to a vertex of low degree reduces the possible colors of few other vertices.
Furthermore, it is relatively unlikely that they will have a colored neighbor either.
Luckily, we can remove many low-degree vertices.

Vertices with two or more neighbors can be removed from the instance trivially:
their set of neighbors can never contain all three possible colors.
Thus, there will always be a color available for these vertices.

Furthermore, we can also limit the number of vertices of degree three:
we use a branching rule to ensure no connected subgraph of degree-three vertices exists that contains nine or more vertices or any cycle.

\begin{lemma}[{\cite[Lemma~20]{Beigel2005}}]
\label{lem:3cycle}
Let $G$ be a \textsc{3-coloring} instance in which some cycle consists only of degree-three vertices. 
Then we can replace $G$ with smaller instances with work factor at most $1.2433$.
\end{lemma}
\begin{proof}
See Beigel and Eppstein\cite[Lemma~20]{Beigel2005}.
\end{proof}

The next lemma has been adapted from Beigel and Eppstein \cite[Lemma~21]{Beigel2005}.
Originally, it removed connected subgraphs of eight or more degree-three vertices, but we edit this to nine or more degree-three vertices.
\begin{lemma}
\label{lem:3subgraph}
Let $G$ be a \textsc{3-coloring} instance containing a connected subgraph of nine or more
degree-three vertices. Then we can replace $G$ with smaller instances with work factor at most $1.3022$.
\end{lemma}
\begin{proof}
Notice that the neighbors of any vertex must be 2-colorable.
So, if a vertex has three neighbors, at least two of its neighbors must receive the same color.
There are three possible pairs of neighbors to select, we will explore each of these in a branch.

When we select two vertices to have the same color, we can merge the two vertices:
the new vertex will neighbor all vertices that either of the two original vertices neighbored.
As the merged vertices neighbored the same vertex of degree three, this vertex will have degree two after the merge.
At this point, this vertex can be removed from the instance.
Furthermore, after removing it, we also remove any of its neighbors of degree three.

Suppose that a connected subgraph of degree-three vertices has $n \geq 9$ vertices.
We will now choose a vertex $v$, whose neighbors we will merge in the three branches.
In case there exists a vertex in the subgraph with three neighbors in the vertex, let any such vertex be $v$.
Otherwise, the subgraph must be a path.
In this case, select a vertex to be $v$, such that removing this vertex from the subgraph causes the subgraph to be split into connected components of size $\lceil \frac{n}{2} \rceil$ and $\lfloor \frac{n}{2} \rfloor$.

\begin{figure}[H]
\centering
\begin{tikzpicture}[main/.style = {draw, circle}, minimum size=0.6cm]
\node[main] (1) {v}; 
\node[main] (2) [below left of=1] {};
\node[main] (3) [left of=2] {};
\node[main] (4) [below of=2] {};
\node[main] (5) [above of=1] {};
\node[main] (6) [above of=5] {};
\node[main] (7) [below right of=1] {};
\node[main] (8) [below of=7] {};
\node[main] (9) [right of=7] {};

\draw[] (1) -- (2);
\draw[] (2) -- (3);
\draw[] (2) -- (4);
\draw[] (1) -- (5);
\draw[] (5) -- (6);
\draw[] (1) -- (7);
\draw[] (7) -- (8);
\draw[] (7) -- (9);

\end{tikzpicture}
\hspace{2cm}%
\begin{tikzpicture}[main/.style = {draw, circle}, minimum size=0.6cm]
\node[main,fill=lightgray] (1) {v}; 
\node[main] (27) [below of=1] {};
\node[main] (3) [left of=2] {};
\node[main] (4) [below of=2] {};
\node[main,fill=lightgray] (5) [above of=1] {};
\node[main,fill=lightgray] (6) [above of=5] {};
\node[main] (8) [below of=7] {};
\node[main] (9) [right of=7] {};

\draw[] (1) -- (27);
\draw[] (27) -- (3);
\draw[] (27) -- (4);
\draw[] (1) -- (5);
\draw[] (5) -- (6);
\draw[] (1) -- (27);
\draw[] (27) -- (8);
\draw[] (27) -- (9);

\end{tikzpicture} 
\caption{A connected subgraph of nine degree-three vertices (left), and one of the three branches which merged two neighbors (right). The bottom neighbors of $v$ are merged. \textcolor{gray}{Gray}: vertices can be removed from the instance.}
\label{fig:ex31}
\end{figure}
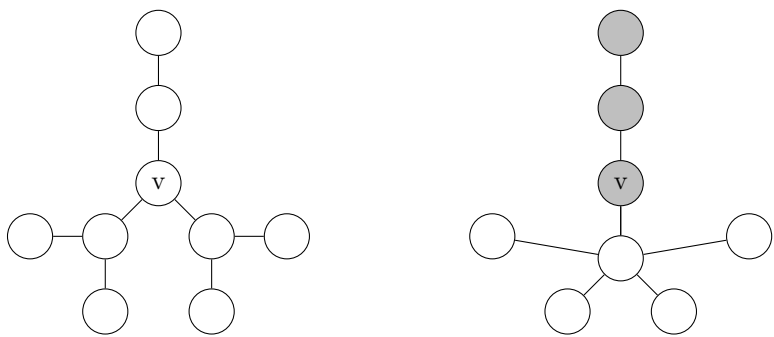

In each branch, $v$ will be removed from the instance and two of its neighbors are merged into one:
each branch reduces the instance size by at least one.
Then, the neighbor of $v$ that was not merged will lose a neighbor and possibly have only two remaining vertices.
In this case, it is removed as well.
We continue this process until no vertices of degree two exist within the subgraph.
There exists exactly one path from each vertex in the subgraph to $v$, so every vertex (besides $v$) can be removed from the instance in exactly one subgraph.
So, we get a work factor $\lambda(2+\alpha_1, 2+\alpha_2, 2+\alpha_3)$, where $\alpha_1+\alpha_2+\alpha_3=n-1$.
Furthermore, if any of $\alpha_1$, $\alpha_2$, or $\alpha_3$ is zero, then the other two must be $\lceil \frac{n}{2} \rceil$ and $\lfloor \frac{n}{2} \rfloor$.

The worst-case values for $\alpha_1$, $\alpha_2$, and $\alpha_3$ are $0$, $4$, and $4$.
We get a work factor of $\lambda(2, 6, 6)=1.3022$.
\end{proof}

\subsection{Coloring Vertices Faster}
After limiting the number of low-degree vertices in the graph, Beigel and Eppstein used a graph structure called a bushy forest to find a small set of vertices with a large set of neighboring vertices.

\begin{definition}[Bushy Forest]
A \textit{bushy forest} is a forest where every tree has at least one internal vertex and each internal vertex must be adjacent to at least four other vertices in the tree.
A bushy forest is maximal, when there does not exist any vertex outside the bushy forest with four neighbors outside the bushy forest, there does not exist any leaf in the bushy forest with three neighbors outside the bushy forest, and there does not exist any vertex outside the bushy forest adjacent to an internal vertex of the bushy forest.
\end{definition}

We find a maximal bushy forest and color the internal vertices of the bushy forest.
All leaves are adjacent to some internal vertex, so all leaves will have a colored neighbor.
Each internal vertex must have at least four neighbors within the bushy forest, so we expect many leaves to exist within the bushy forest.
To analyze the number of possible color assignments, we partition the vertices of the graph based on their relation to the maximal bushy forest:

\begin{itemize}
\item $R$: the root vertices of the trees in the bushy forest, where for each tree one arbitrary internal vertex is chosen to be the root vertex.
\item $I$: all internal vertices of the bushy forest that are not rood vertices.
\item $L$: the leaves of the bushy forest.
\item $N$: vertices outside the bushy forest that neighbor the bushy forest.
\item $U$: vertices outside the bushy forest that do not neighbor the bushy forest.
\end{itemize}

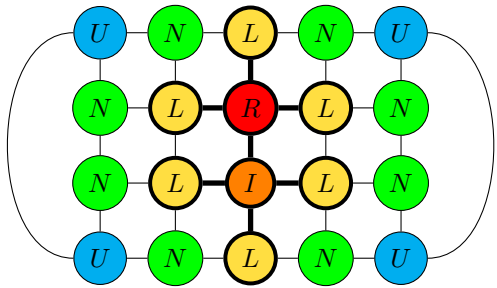
\begin{figure}[H]
\centering
\begin{tikzpicture}[main/.style = {draw, circle}, node distance=1cm, minimum size=0.6cm]
\node[main,fill=red,line width=1.5pt] (i1) {$R$}; 
\node[main,fill=orange,line width=1.5pt] (i2) [below of=1] {$I$};

\node[main,fill=Goldenrod,line width=1.5pt] (l1) [left of=i1] {$L$};
\node[main,fill=Goldenrod,line width=1.5pt] (l2) [above of=i1] {$L$};
\node[main,fill=Goldenrod,line width=1.5pt] (l3) [right of=i1] {$L$};
\node[main,fill=Goldenrod,line width=1.5pt] (l4) [right of=i2] {$L$};
\node[main,fill=Goldenrod,line width=1.5pt] (l5) [below of=i2] {$L$};
\node[main,fill=Goldenrod,line width=1.5pt] (l6) [left of=i2] {$L$};

\node[main,fill=green] (u1) [left of=l1] {$N$};
\node[main,fill=green] (u2) [left of=l2] {$N$};
\node[main,fill=green] (u3) [left of=l5] {$N$};
\node[main,fill=green] (u4) [left of=l6] {$N$};
\node[main,fill=green] (u5) [right of=l3] {$N$};
\node[main,fill=green] (u6) [right of=l2] {$N$};
\node[main,fill=green] (u7) [right of=l5] {$N$};
\node[main,fill=green] (u8) [right of=l4] {$N$};
\node[main,fill=cyan] (u9) [left of=u2] {$U$};
\node[main,fill=cyan] (u10) [left of=u3] {$U$};
\node[main,fill=cyan] (u11) [right of=u6] {$U$};
\node[main,fill=cyan] (u12) [right of=u7] {$U$};

\draw[line width=2.0pt] (i1) -- (i2);
\draw[line width=2.0pt] (i1) -- (l1);
\draw[line width=2.0pt] (i1) -- (l2);
\draw[line width=2.0pt] (i1) -- (l3);
\draw[line width=2.0pt] (i2) -- (l4);
\draw[line width=2.0pt] (i2) -- (l5);
\draw[line width=2.0pt] (i2) -- (l6);

\draw[] (l3) -- (l4);
\draw[] (l1) -- (l6);

\draw[] (l1) -- (u1);
\draw[] (l1) -- (u2);
\draw[] (l2) -- (u2);
\draw[] (l2) -- (u6);
\draw[] (l3) -- (u6);
\draw[] (l3) -- (u5);
\draw[] (l4) -- (u8);
\draw[] (l4) -- (u7);
\draw[] (l5) -- (u7);
\draw[] (l5) -- (u3);
\draw[] (l6) -- (u3);
\draw[] (l6) -- (u4);

\draw[] (u1) -- (u4);
\draw[] (u5) -- (u8);
\draw[] (u2) -- (u9);
\draw[] (u1) -- (u9);
\draw[] (u3) -- (u10);
\draw[] (u4) -- (u10);
\draw[] (u6) -- (u11);
\draw[] (u5) -- (u11);
\draw[] (u7) -- (u12);
\draw[] (u8) -- (u12);
\draw (u9) to [out=180, in=180] (u10);
\draw (u11) to [out=0, in=0] (u12);

\end{tikzpicture}
\caption{A maximal bushy forest (in bold) consisting of a single tree in a graph. \textcolor{red}{Red}: root vertices ($R$). \textcolor{orange}{Orange}: internal vertices ($I$). \textcolor{Goldenrod}{Yellow}: leaves ($L$). \textcolor{green}{Green}: neighbors to the bushy forest ($N$). \textcolor{cyan}{Cyan}: other vertices ($U$).}
\label{fig:exbf}
\end{figure}

Given a graph with any maximal bushy forest, we use these partitions to analyze the runtime of the algorithm.
For $R$ and $I$, the internal vertices of $F$, we try every possible color assignment. 
Every root vertex has three possible colors, so there are $3^{|R|}$ possible color assignments for vertices in $|R|$.
Afterward, every tree will have its root vertex colored.
Vertices in $I$ adjacent to $R$ have two available colors.
We color every vertex in $I$ after a neighbor has been colored: there are only $2^{|I|}$ possible color assignments for vertices in $I$.
Then, every vertex in $L$ will have at least one colored neighbor, meaning that these can be removed from the instance in polynomial time.
This leaves vertices in $N$ and $U$ as the vertices that need to be solved by the \textsc{(3,2)-Constraint Satisfaction Problem}: this takes time $\mathcal{O}^*(1.36443^{|N|+|U|})$. 
Overall, it takes time $\mathcal{O}^*(3^{|R|}\cdot2^{|I|}\cdot1.36443^{|N|+|U|})$ to solve \textsc{3-coloring} given a bushy forest $F$.

Luckily, we can also color some vertices in $N$ and $U$ to improve the running time of the algorithm.
To do this, we find another forest over the vertices $V$ not covered by the bushy forest $F$.
We denote this graph as $G[V-F]$: the induced subgraph of $V-F$ in $G$.

\begin{definition}[Chromatic Forest]
    A forest of rooted trees, where the root of each tree has exactly three children and at most five grandchildren. Each child of the root has at most two children itself.
\end{definition}

Beigel and Eppstein showed that there must exist a chromatic forest that covers all vertices in $U$, a \textit{maximal chromatic forest}. 
Beigel and Eppstein proved that the following algorithm will always find a valid maximal chromatic forest. 
Here, a $K_{1,3}$ tree is a tree with a root vertex and three children.

\begin{enumerate}
    \item Create a maximal forest of $K_{1,3}$ trees.
    \item While possible, remove a $K_{1,3}$ tree and add two new ones.
    \item Assign all remaining vertices in $U$ as a grandchild to some $K_{1,3}$ tree, such that no tree has six grandchildren.
\end{enumerate}

Beigel and Eppstein's analysis of the trees in the chromatic forest was slightly different from our analysis.
They expressed the number of color assignments in the number of degree-three vertices.
Instead, we express the number of color assignments in the number of vertices in $G[V-F]$.
We do this, such that we can further improve the running time by proving that some vertices with more than three neighbors must be included in the maximal chromatic forest.

For the next lemma, see Beigel and Eppstein for a similar proof \cite[Lemma~24]{Beigel2005}.
There, they express the runtime to color a tree in a chromatic forest in the number of degree-three vertices in the tree.
Furthermore, they also assume that there can be at most eight vertices in a connected subgraph of degree-three vertices.

\begin{lemma}
\label{lem:chromaticforest}
Let $T'$ be a tree in a chromatic forest. Then $T'$ can be colored with a work factor of $1.34004$.
\end{lemma}
\begin{proof}
Presume $T$ has at most four grandchildren. 
In this case, we give the root one of the three possible colors. 
This removes one color as a possibility for its three children, allowing those to be removed from the instance by the \textsc{(3,2)-CSP} algorithm. 
The grandchildren remain and will have to be solved by the time $\mathcal{O}^*(1.36443^n)$ algorithm. 
As every grandchild is colored by the $\mathcal{O}^*(1.36443^n)$ algorithm, they are relatively slow to color. 
The worst-case scenario happens when there are as many grandchildren as possible: four.
We iterate over all three possible color assignments of the root vertex, while four vertices have to be colored by the \textsc{(3,2)-CSP} algorithm. On average, the runtime required per vertex in the tree to color all vertices in the tree is $(3\cdot 1.36443^4)^\frac{1}{8} < 1.34004$.

Now, let $T$ be a tree with exactly five grandchildren.
Vertices in the chromatic forest have at most three neighbors, so every child has at most two grandchildren.
There are five grandchildren, so two children have two grandchildren, while the other child has one grandchild. 
We select the two children with two grandchildren and iterate over all color assignments for these vertices. 
If they receive the same color, the root, and their combined four grandchildren will all have a colored neighbor.
This means that only the remaining other child and grandchild are colored through the \textsc{(3,2)-Constraint Satisfaction Problem}. If they receive different colors, the color of the root must be the third color.
Then, only the last grandchild is colored by the \textsc{(3,2)-CSP} algorithm. 
This results in a runtime per vertex of $(3\cdot 1.36443^2 + 6\cdot 1.36443)^\frac{1}{9} < 1.338302 < 1.34004$.

Thus, the worst case is a tree with four grandchildren, which leads to the claim of the lemma that in the worst case, we can color a tree $T'$ with work factor $1.34004$.
\end{proof}

We have summarized all of Beigel and Eppstein's concepts that are important for our algorithm.
We remove vertices with fewer than three neighbors from the graph; and connected subgraphs of degree-three vertices that contain a cycle or nine or more vertices. 
Afterward, we find a maximal bushy forest and a maximal chromatic forest to find a set of vertices with many neighbors.
We color specific vertices in the two forests, which allows us to remove their neighbors from the instance.
The remaining vertices will be colored using the \textsc{(3,2)-CSP} algorithm.


\section{Limiting Difficult-To-Color Vertices}
\label{sec:forest}
We will now discuss our main improvement: 
a modification to the bushy forest that further restricts the number of vertices that can exist outside the bushy forest. 
We will limit the existence of high-magnitude vertices in the set $N$.
This vertex enables many vertices to exist outside the bushy forest, which we want to avoid.

\begin{definition}[High-Magnitude Vertex]
A vertex adjacent to a maximal bushy forest is called a \textit{high-magnitude vertex} if it has three neighbors outside the maximal bushy forest.
\end{definition}
\begin{definition}[Maximal Low-Magnitude Bushy Forest]
A maximal bushy forest is a \textit{maximal low-magnitude bushy forest} where all the adjacent high-magnitude vertices are adjacent to a tree with a single internal vertex and four leaves. 
Every two high-magnitude vertices adjacent to the same tree must share a common neighbor. 
This neighbor must either be the leaf of this tree, or a vertex outside the bushy forest.
\end{definition}

We can transform any given maximal bushy forest into a maximal low-magnitude bushy forest in polynomial time.
If there exists a high-magnitude vertex that causes the bushy forest to not be a low-magnitude bushy forest, we can modify the tree adjacent to the high-magnitude vertex.
Every such modification either adds more internal vertices or adds a new tree to the bushy forest.
Furthermore, no modification will decrease the number of trees in the bushy forest.
The bushy forest can never include more trees or internal vertices than the number of vertices in the graph, we will find a maximal low-magnitude bushy forest.
In case a modification causes the bushy forest to not be maximal anymore, we add a new internal vertex to the bushy forest.
This will also increase the number of internal vertices.

\begin{lemma}
Let $G$ be a graph in which all vertex degrees are three or more, in which there is no cycle of degree-three vertices nor any connected subgraph of nine or more degree-three vertices. Then there exists a maximal low-magnitude bushy forest, which we can find in polynomial time.
\end{lemma}
\begin{proof}
First, we find a maximal bushy forest in polynomial time: greedily add new trees to the bushy forest, and add new internal vertices to trees, until this is no longer possible.
We will now modify this maximal bushy forest to find a maximal low-magnitude bushy forest. 
Specifically, we will present cases where there exists a high-magnitude vertex adjacent to a tree with multiple internal vertices, more than four leaves, or a tree with a pair of high-magnitude vertices that do not share a common neighbor.

\begin{figure}[H]
\centering
\begin{tikzpicture}[main/.style = {draw, circle}, node distance=1cm, minimum size=0.6cm]
\node[main,line width=1.5pt,fill=Goldenrod] (1) {l}; 
\node[main,line width=1.5pt,fill=orange] (2) [below of=1] {};
\node[main,line width=1.5pt,fill=red] (3) [left of=2] {r};
\node[main,line width=1.5pt,fill=Goldenrod] (4) [below of=3] {};
\node[main,line width=1.5pt,fill=Goldenrod] (5) [below of=2] {};
\node[main,line width=1.5pt,fill=Goldenrod] (6) [left of=3] {};
\node[main,line width=1.5pt,fill=Goldenrod] (7) [right of=2] {};
\node[main,line width=1.5pt,fill=Goldenrod] (8) [above of=3] {};
\node[main] (9) [above of=1,fill=lime] {v};
\node[main] (10) [above of=9,fill=cyan] {};
\node[main] (11) [left of=9,fill=cyan] {};
\node[main] (12) [right of=9,fill=cyan] {};

\draw[line width=2.0pt] (1) -- (2);
\draw[line width=2.0pt] (2) -- (3);
\draw[line width=2.0pt] (2) -- (5);
\draw[line width=2.0pt] (2) -- (7);
\draw[line width=2.0pt] (3) -- (4);
\draw[line width=2.0pt] (3) -- (6);
\draw[line width=2.0pt] (3) -- (8);
\draw (1) -- (9);
\draw (9) -- (10);
\draw (9) -- (11);
\draw (9) -- (12);

\end{tikzpicture}
\hspace{1cm}%
\begin{tikzpicture}[main/.style = {draw, circle}, node distance=1cm, minimum size=0.6cm]
\node[main,line width=1.5pt,fill=Goldenrod] (1) {l}; 
\node[main,line width=1.5pt,fill=Goldenrod] (2) [below of=1] {};
\node[main,line width=1.5pt,fill=red] (3) [left of=2] {r};
\node[main,line width=1.5pt,fill=Goldenrod] (4) [below of=3] {};
\node[main,fill=green] (5) [below of=2] {};
\node[main,line width=1.5pt,fill=Goldenrod] (6) [left of=3] {};
\node[main,fill=green] (7) [right of=2] {};
\node[main,line width=1.5pt,fill=Goldenrod] (8) [above of=3] {};
\node[main,line width=1.5pt] (9) [above of=1,fill=red] {v};
\node[main,line width=1.5pt] (10) [above of=9,fill=Goldenrod] {};
\node[main,line width=1.5pt] (11) [left of=9,fill=Goldenrod] {};
\node[main,line width=1.5pt] (12) [right of=9,fill=Goldenrod] {};

\draw (1) -- (2);
\draw[line width=2.0pt] (2) -- (3);
\draw (2) -- (5);
\draw (2) -- (7);
\draw[line width=2.0pt] (3) -- (4);
\draw[line width=2.0pt] (3) -- (6);
\draw[line width=2.0pt] (3) -- (8);
\draw[line width=2.0pt] (1) -- (9);
\draw[line width=2.0pt] (9) -- (10);
\draw[line width=2.0pt] (9) -- (11);
\draw[line width=2.0pt] (9) -- (12);

\end{tikzpicture}
\caption{A graph with a bushy forest (in bold) displaying the situation before (left) and after (right) removing high-magnitude vertices adjacent to large trees. \textcolor{red}{Red}: root vertices ($R$).
\textcolor{orange}{Orange}: internal vertices ($I$). \textcolor{Goldenrod}{Yellow}: leaves ($L$). \textcolor{lime}{Lime}: high-magnitude vertices. \textcolor{green}{Green}: other neighbors to the bushy forest ($N$). \textcolor{cyan}{Cyan}: other vertices ($U$).}
\label{fig:ex41}
\end{figure}
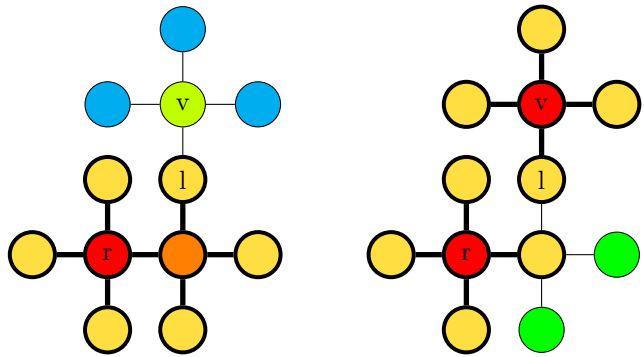

Consider a tree consisting of multiple internal vertices with an adjacent high-magnitude vertex $v$, as shown in Figure \ref{fig:ex41}. 
Let the leaf adjacent to $v$ be $l$ and let the root vertex $r$ of the tree be any internal vertex that is not adjacent to $l$ in the bushy forest.
Then, we can remove the original tree, and add two trees rooted at $r$ and $v$ respectively.
The tree rooted at $r$ will include four vertices from the original tree, while the tree rooted at $v$ will include its three neighbors outside the bushy forest and its neighboring leaf $l$.

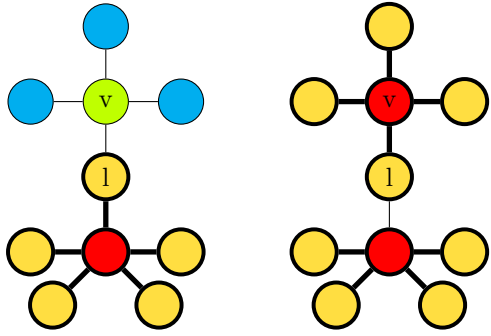
\begin{figure}[H]
\centering
\begin{tikzpicture}[main/.style = {draw, circle}, node distance=1cm, minimum size=0.6cm]
\node[main,line width=1.5pt,fill=Goldenrod] (1) {l}; 
\node[main,line width=1.5pt,fill=red] (2) [below of=1] {};
\node[main,line width=1.5pt,fill=Goldenrod] (3) [below right of=2] {};
\node[main,line width=1.5pt,fill=Goldenrod] (4) [below left of=2] {};
\node[main,line width=1.5pt,fill=Goldenrod] (5) [left of=2] {};
\node[main,line width=1.5pt,fill=Goldenrod] (7) [right of=2] {};
\node[main] (9) [above of=1,fill=lime] {v};
\node[main] (10) [above of=9,fill=cyan] {};
\node[main] (11) [left of=9,fill=cyan] {};
\node[main] (12) [right of=9,fill=cyan] {};

\draw[line width=2.0pt] (1) -- (2);
\draw[line width=2.0pt] (2) -- (3);
\draw[line width=2.0pt] (2) -- (5);
\draw[line width=2.0pt] (2) -- (7);
\draw[line width=2.0pt] (2) -- (4);
\draw (1) -- (9);
\draw (9) -- (10);
\draw (9) -- (11);
\draw (9) -- (12);

\end{tikzpicture}
\hspace{1cm}%
\begin{tikzpicture}[main/.style = {draw, circle}, node distance=1cm, minimum size=0.6cm]
\node[main,line width=1.5pt,fill=Goldenrod] (1) {l}; 
\node[main,line width=1.5pt,fill=red] (2) [below of=1] {};
\node[main,line width=1.5pt,fill=Goldenrod] (3) [below right of=2] {};
\node[main,line width=1.5pt,fill=Goldenrod] (4) [below left of=2] {};
\node[main,line width=1.5pt,fill=Goldenrod] (5) [left of=2] {};
\node[main,line width=1.5pt,fill=Goldenrod] (7) [right of=2] {};
\node[main,line width=1.5pt,fill=red] (9) [above of=1] {v};
\node[main,line width=1.5pt,fill=Goldenrod] (10) [above of=9] {};
\node[main,line width=1.5pt,fill=Goldenrod] (11) [left of=9] {};
\node[main,line width=1.5pt,fill=Goldenrod] (12) [right of=9] {};

\draw (1) -- (2);
\draw[line width=2.0pt] (2) -- (3);
\draw[line width=2.0pt] (2) -- (5);
\draw[line width=2.0pt] (2) -- (7);
\draw[line width=2.0pt] (2) -- (4);
\draw[line width=2.0pt] (1) -- (9);
\draw[line width=2.0pt] (9) -- (10);
\draw[line width=2.0pt] (9) -- (11);
\draw[line width=2.0pt] (9) -- (12);

\end{tikzpicture} 
\caption{A graph with a bushy forest (in bold) displaying the situation before (left) and after (right) removing high-magnitude vertices adjacent to trees with many leaves. \textcolor{red}{Red}: root vertices ($R$). \textcolor{Goldenrod}{Yellow}: leaves ($L$). \textcolor{lime}{Lime}: high-magnitude vertices. \textcolor{cyan}{Cyan}: other vertices ($U$).}
\label{fig:ex42}
\end{figure}

Secondly, consider a high-magnitude vertex $v$ adjacent to a tree with a single internal vertex and at least five leaves, like shown in Figure \ref{fig:ex42}. 
Let the leaf adjacent to $v$ be called $l$. 
Remove $l$ from the existing tree, which will still be a valid tree with four leaves. 
Then, add a new tree rooted at $v$ with four leaves:
the three neighbors originally outside the bushy forest and $l$.

\begin{figure}[H]
\centering
\begin{tikzpicture}[main/.style = {draw, circle}, node distance=1cm, minimum size=0.6cm]
\node[main,line width=1.5pt,fill=Goldenrod] (1) {}; 
\node[main,line width=1.5pt,fill=red] (2) [below of=1] {};
\node[main,line width=1.5pt,fill=Goldenrod] (3) [below of=2] {};
\node[main,line width=1.5pt,fill=Goldenrod] (4) [left of=2] {};
\node[main,line width=1.5pt,fill=Goldenrod] (5) [right of=2] {};
\node[main,fill=lime] (6) [below left of=3] {w};
\node[main,fill=cyan] (7) [below left of=6] {};
\node[main,fill=cyan] (8) [left of=6] {};
\node[main,fill=cyan] (9) [above left of=6] {};
\node[main,fill=lime] (10) [above left of=1] {v};
\node[main,fill=cyan] (11) [above left of=10] {};
\node[main,fill=cyan] (12) [left of=10] {};
\node[main,fill=cyan] (13) [below left of=10] {};

\draw[line width=2.0pt] (1) -- (2);
\draw[line width=2.0pt] (2) -- (3);
\draw[line width=2.0pt] (2) -- (4);
\draw[line width=2.0pt] (2) -- (5);
\draw (3) -- (6);
\draw (6) -- (7);
\draw (6) -- (8);
\draw (6) -- (9);
\draw (1) -- (10);
\draw (10) -- (11);
\draw (10) -- (12);
\draw (10) -- (13);

\end{tikzpicture}
\hspace{1cm}%
\begin{tikzpicture}[main/.style = {draw, circle}, node distance=1cm, minimum size=0.6cm]
\node[main,line width=1.5pt,fill=Goldenrod] (1) {}; 
\node[main,fill=green] (2) [below of=1] {};
\node[main,line width=1.5pt,fill=Goldenrod] (3) [below of=2] {};
\node[main,fill=cyan] (4) [left of=2] {};
\node[main,fill=cyan] (5) [right of=2] {};
\node[main,line width=1.5pt,fill=red] (6) [below left of=3] {w};
\node[main,line width=1.5pt,fill=Goldenrod] (7) [below left of=6] {};
\node[main,line width=1.5pt,fill=Goldenrod] (8) [left of=6] {};
\node[main,line width=1.5pt,fill=Goldenrod] (9) [above left of=6] {};
\node[main,line width=1.5pt,fill=red] (10) [above left of=1] {v};
\node[main,line width=1.5pt,fill=Goldenrod] (11) [above left of=10] {};
\node[main,line width=1.5pt,fill=Goldenrod] (12) [left of=10] {};
\node[main,line width=1.5pt,fill=Goldenrod] (13) [below left of=10] {};

\draw (1) -- (2);
\draw (2) -- (3);
\draw (2) -- (4);
\draw (2) -- (5);
\draw[line width=2.0pt] (3) -- (6);
\draw[line width=2.0pt] (6) -- (7);
\draw[line width=2.0pt] (6) -- (8);
\draw[line width=2.0pt] (6) -- (9);
\draw[line width=2.0pt] (1) -- (10);
\draw[line width=2.0pt] (10) -- (11);
\draw[line width=2.0pt] (10) -- (12);
\draw[line width=2.0pt] (10) -- (13);

\end{tikzpicture}
\caption{A graph with a bushy forest (in bold) displaying the situation before (left) and after (right) removing high-magnitude vertices adjacent to the same tree. \textcolor{red}{Red}: root vertices ($R$). \textcolor{Goldenrod}{Yellow}: leaves ($L$). \textcolor{lime}{Lime}: high-magnitude vertices. \textcolor{green}{Green}: other neighbors to the bushy forest ($N$). \textcolor{cyan}{Cyan}: other vertices ($U$).}
\label{fig:ex43}
\end{figure}
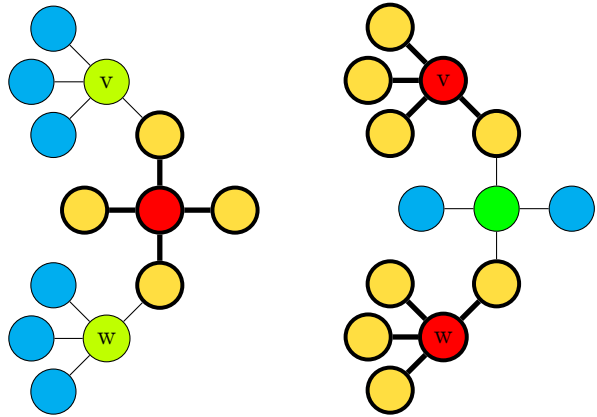

Next, consider a tree adjacent with two distinct leaves adjacent to high-magnitude vertices $v$ and $w$, where $v$ and $w$ are neither adjacent nor share a neighbor outside the bushy forest, as shown in Figure \ref{fig:ex43}.
Then, we can remove this tree, and instead root new trees at $v$ and $w$.
Both new trees will have four leaves: 
the three neighbors originally outside the bushy forest and the leaves to which they were adjacent. 
As $v$ and $w$ are adjacent to distinct leaves and do not share a neighbor outside the bushy forest, both will form valid trees.

\begin{figure}[H]
\centering
\begin{tikzpicture}[main/.style = {draw, circle}, node distance=1cm, minimum size=0.6cm]
\node[main,line width=1.5pt,fill=red] (1) {}; 
\node[main,line width=1.5pt,fill=Goldenrod] (2) [below of=1] {};
\node[main,line width=1.5pt,fill=Goldenrod] (3) [above of=1] {};
\node[main,line width=1.5pt,fill=Goldenrod] (4) [left of=1] {};
\node[main,line width=1.5pt,fill=Goldenrod] (5) [right of=1] {};
\node[main,fill=lime] (6) [above right of=3] {v};
\node[main,fill=lime] (7) [above right of=5] {w};
\node[main,fill=cyan] (8) [above right of=6] {};
\node[main,fill=cyan] (9) [above right of=7] {};
\node[main,fill=cyan] (10) [above left of=6] {};
\node[main,fill=cyan] (11) [below right of=7] {};

\draw[line width=2.0pt] (1) -- (2);
\draw[line width=2.0pt] (1) -- (3);
\draw[line width=2.0pt] (1) -- (4);
\draw[line width=2.0pt] (1) -- (5);
\draw (3) -- (6);
\draw (5) -- (7);
\draw (6) -- (7);
\draw (6) -- (8);
\draw (6) -- (10);
\draw (7) -- (9);
\draw (7) -- (11);

\end{tikzpicture}
\hspace{1cm}%
\begin{tikzpicture}[main/.style = {draw, circle}, node distance=1cm, minimum size=0.6cm]
\node[main,fill=green] (1) {}; 
\node[main,fill=cyan] (2) [below of=1] {};
\node[main,line width=1.5pt,fill=Goldenrod] (3) [above of=1] {};
\node[main,fill=cyan] (4) [left of=1] {};
\node[main,line width=1.5pt,fill=Goldenrod] (5) [right of=1] {};
\node[main,line width=1.5pt,fill=red] (6) [above right of=3] {v};
\node[main,line width=1.5pt,fill=red] (7) [above right of=5] {w};
\node[main,line width=1.5pt,fill=Goldenrod] (8) [above right of=6] {};
\node[main,line width=1.5pt,fill=Goldenrod] (9) [above right of=7] {};
\node[main,line width=1.5pt,fill=Goldenrod] (10) [above left of=6] {};
\node[main,line width=1.5pt,fill=Goldenrod] (11) [below right of=7] {};

\draw (1) -- (2);
\draw (1) -- (3);
\draw (1) -- (4);
\draw (1) -- (5);
\draw[line width=2.0pt] (3) -- (6);
\draw[line width=2.0pt] (5) -- (7);
\draw[line width=2.0pt] (6) -- (7);
\draw[line width=2.0pt] (6) -- (8);
\draw[line width=2.0pt] (6) -- (10);
\draw[line width=2.0pt] (7) -- (9);
\draw[line width=2.0pt] (7) -- (11);

\end{tikzpicture}
\caption{A graph with a bushy forest (in bold) displaying the situation before (left) and after (right) removing high-magnitude vertices adjacent to large trees. \textcolor{red}{Red}: root vertices ($R$). \textcolor{Goldenrod}{Yellow}: leaves ($L$). \textcolor{lime}{Lime}: high-magnitude vertices. \textcolor{green}{Green}: other neighbors to the bushy forest ($N$). \textcolor{cyan}{Cyan}: other vertices ($U$).}
\label{fig:ex44}
\end{figure}
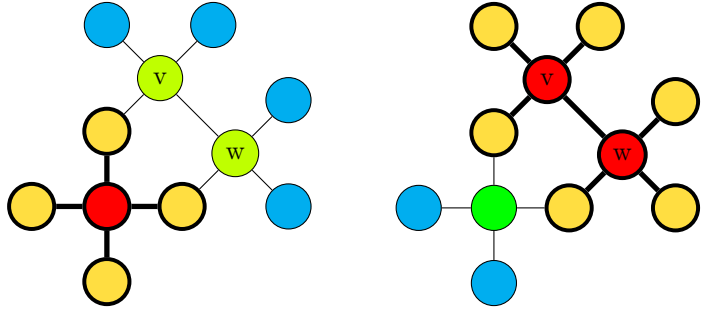

Finally, consider a tree with two distinct leaves adjacent to high-magnitude vertices $v$ and $w$, where $v$ and $w$ do not share a common neighbor outside the bushy forest, but are adjacent.
Then, we can remove this tree (which has a singular internal vertex) and instead create a new tree with both high-magnitude $v$ and $w$ as internal vertices, as shown in Figure \ref{fig:ex44}.

Thus, a high-magnitude vertex can only be adjacent to a tree in the maximal low-magnitude bushy forest with exactly one internal vertex and four leaves.
If there are multiple high-magnitude vertices adjacent to the same tree, they must either be adjacent to the same leaf of the tree or share a common neighbor outside the bushy forest.
\end{proof}

Now, any remaining high-magnitude vertices must be adjacent to a low-magnitude bushy forest, which constrains their appearance.
Furthermore, these high-magnitude vertices must share a common neighbor with any other high-magnitude vertices adjacent to the same tree.
As they have to share a common neighbor, this creates predictable structures within the graph when one tree is adjacent to many high-magnitude vertices.
In the next section, we explain how these structures also limit the number of vertices outside the bushy forest.


\section{A Bound on Difficult-To-Color Vertices}
\label{sec:analysis}
High-magnitude vertices only exist adjacent to certain trees, as described in the previous section. 
High-magnitude vertices must share a common neighbor if they are adjacent to the same tree.
Using these constraints, we find a set of equations expressing the relationship between the different types of vertices.
We will use these constraints to formulate a linear program that finds the worst-case graph for our algorithm.

We further partition the vertices outside the bushy forest, such that we can analyze the effects of high-magnitude vertices:
\begin{itemize}
\item $N_1$: vertices in $N$ of degree three, with one neighbor in $L$.
\item $N_2$: vertices in $N$ with multiple neighbors in $L$.
\item $N_{3,i}$: high-magnitude vertices (vertices in $N$ with three neighbors outside the maximal low-magnitude bushy forest), adjacent to a tree to which in total $i$ high-magnitude vertices are adjacent.
If a high-magnitude vertex is adjacent to multiple trees in the bushy forest, we assign them to the highest value $i$ of its adjacent trees.
We refer to these vertices collectively as $N_3$.
\end{itemize}

Secondly, we will partition the vertices in $U$ to express their relation to the new partition of $N$:
\begin{itemize}
    \item $U'$: vertices in $U$ of which all three neighbors must be in $N_3$.
    \item $U_j$: vertices in $U$ (and not in $U'$) that appear in a connected subgraph of degree-three vertices containing $j$ vertices in $N_1$.
\end{itemize}

All vertices in $U$ must have degree three, so they must all appear in exactly one subgraph of degree-three vertices.
Furthermore, the connected subgraphs of degree-three vertices can contain at most eight vertices.
As such, all vertices in $U$ will either be in $U'$ or some set $U_j$ for $0 \geq j < 8$.

Finally, we notice that the number of vertices in $N_2$ and $N_3$ determines the maximum amount of degree-three vertices outside the bushy forest ($N_1$ and all vertices in $U$).
We count the edges that can exist between vertices in $N_2$ or $N_3$ and degree-three vertices outside the bushy forest. We split this into two categories:
\begin{enumerate}
    \item Three edges per vertex in $U'$ to connect $N_3$ with $U'$.
    \item The set of edges $S$ between vertices in $N_2$ or $N_3$ and vertices $N_1$ or a set $U_j$ where $0 \leq j \leq 7$. 
\end{enumerate}

Now, we will show various constraints based on the relation between the partition of vertices. These relations will limit how many vertices in one partition can exist, based on the number of vertices in another partition:

\begin{lemma}
Let $G$ be a graph in which all vertex degrees are three or more, in which there is no cycle of degree-three vertices nor any connected subgraph of nine or more degree-three vertices. Let $F$ be a maximal low-magnitude bushy forest in $G$. Then,
\end{lemma}
\begin{align}
    4\cdot|R|+2\cdot|I|\leq|L| \label{align:1}\\
    |N_1|+2\cdot|N_2|+|N_3|\leq 2\cdot|L| \label{align:2}\\
    |U'|\geq\frac{1}{5}\cdot|N_{3,5}|+\frac{2}{6}\cdot|N_{3,6}|+\frac{5}{7}\cdot|N_{3,7}|+|N_{3,8}|\label{align:3}\\
    \sum_{j=0}^7 \left(\frac{10-j}{8-j}\cdot |U_j|\right)\leq 2\cdot |N_2|+  3\cdot\sum_{i=1}^8|N_{3,i}|-3\cdot|U'|\label{align:4}
\end{align}

\begin{proof}
The equations describe the relationship between the partitions of the vertices.
We will prove the equations one by one:
\[ 4\cdot|R|+2\cdot|I|\leq|L| \]
This follows the definition of the bushy forest: 
If a tree has a single internal vertex, it must have at least four leaves.
Every additional internal vertex causes the tree to have at least two more leaves:
we turn a leaf into an internal vertex.
This removes one of the leaves.
However, the internal vertex must now neighbor three leaves, as it only bordered the internal vertex as a leaf.
Thus, each additional internal vertex causes at least two leaves to appear in the bushy forest.

\[|N_1|+2\cdot|N_2|+|N_3|\leq 2\cdot|L|\]
This follows from the definitions $L$, $N_1$, $N_2$, and $N_3$:
vertices in $L$ have at most two neighbors outside the bushy forest.
Otherwise, the bushy forest would not be maximal.
Vertices in $N_1$ and $N_3$ have at least one neighbor in $L$, while vertices in $N_2$ must have at least two neighbors in $L$.

\[|U'|\geq\frac{1}{5}\cdot|N_{3,5}|+\frac{2}{6}\cdot|N_{3,6}|+\frac{5}{7}\cdot|N_{3,7}|+|N_{3,8}|\]
Recall the definition of $U'$: vertices in $U$ of which all neighbors must be high-magnitude vertices ($N_3$).
Vertices appear in $U'$ when there are five or more high-magnitude vertices adjacent to one tree of the bushy forest.
All high-magnitude vertices must share a common neighbor.
Then, some vertices in $U$ must be adjacent to three of the high-magnitude vertices if all must share a common neighbor.
We analyze how many vertices in $U'$ must exist for each tree in the bushy forest adjacent to five, six, seven, or eight high-magnitude vertices.

Consider a tree with five adjacent high-magnitude vertices (five vertices in $N_{3,5}$).
At most, two pairs of high-magnitude vertices have a common neighbor through a leaf of the tree.
So, at least one vertex does not share a common neighbor with any of the other four vertices through a leaf.
This vertex must share a common neighbor outside the bushy forest with all four other high-magnitude vertices.
However, it can only have three neighbors outside the bushy forest.
So, at least one of its neighbors must border two of the other four high-magnitude vertices: 
it has three neighbors in $N_3$ and must be in $U'$.
For every five vertices in $N_{3,5}$, there must exist at least one vertex in $U'$.

Vertices in $N_1$ and $N_2$ can never have three neighbors in $N_3$, so they cannot ever take the role of a vertex in $U'$.
Vertices in $N_3$ adjacent to a different tree can have three neighbors in $N_3$.
However, this has a strictly better runtime than a vertex in $U$ taking the same role:
it removes a possible vertex in $U$ from the instance.
Vertices in $U$ are colored through the chromatic forest, which takes more time than the overall algorithm.
As such, we can assume that it is always a vertex in $U$ that must neighbor three vertices in $N_3$.

Using this technique, we determine that there must be at least two vertices in $U'$ for every six vertices in $N_{3,6}$, five vertices in $U'$ for every seven vertices in $N_{3,7}$, and eight vertices in $U'$ for every eight vertices in $N_{3,8}$.

\[\sum_{j=0}^7 \left(\frac{10-j}{8-j}\cdot |U_j|\right)\leq 2\cdot |N_2|+  3\cdot|N_3|-3\cdot|U'|\]

Recall $S$: the set of edges between $N_2$ and $N_3$, and $N_1$ and $U$ (but not $U'$).
We will count the number of edges $S$ that may exist per vertex in $N_2$ and $N_3$, and the number of vertices in $U$ (but not $U'$) that may exist per edge in $S$ to prove this equation.

By definition, $N_2$ can have at most two neighbors outside the bushy forest, while $N_3$ will have exactly three.
However, we subtract all edges to vertices in $U'$.
Each vertex in $U'$ has three neighbors in $N_3$, so we subtract three edges that cannot go from $N_3$ to a different vertex in $U$ for each vertex in $U'$.
We get $|S| \leq 2\cdot |N_2|+  3\cdot|N_3|-3\cdot|U'|$.

Then, we count the number of vertices that may exist in $U-U'$ (in $U$, but not in $U'$) per edge in $S$. 
First, we will consider the connected components in $G[(U-U') \cup N_1$: 
these vertices all have degree three, and cannot contain cycles as per Lemma \ref{lem:3cycle}. 
Thus, they form a tree, which means that a connected component of $n$ vertices includes exactly $n-1$ edges.
Each edge has two endpoints within the component.
The sum of degrees of these vertices is $3\cdot n$, so the number of edges with one endpoint in this component and one outside this component must be $3\cdot n - 2\cdot (n-1)=n+2$. 
The ratio of vertices to outgoing edges is $n/(n+2)$, which is maximized when $n$ is maximal: $n=8$.
All outgoing edges having an endpoint in $U-U'$ must have the other endpoint in $N_2$ or $N_3$ by definition.
However, every vertex in $N_1$ will have exactly one neighbor in $L$.
So, for every vertex in $N_1$ in the component, we subtract one neighbor of the component that must be in $N_2$ or $N_3$.
Furthermore, we calculated how many vertices $U_j$ can exist:
we must also subtract the number of vertices in $N_1$ from this.
Then, we get $\sum_{j=0}^7 (\frac{10-j}{8-j}\cdot |U_j|) \leq |S|$.
If we combine these two results, we get $\sum_{j=0}^7 (\frac{10-j}{8-j}\cdot |U_j|)\leq 2\cdot |N_2|+  3\cdot|N_3|-3\cdot|U'|$.
\end{proof}

Next, we show that we can cover all high-magnitude vertices adjacent to a vertex in $U'$ in a maximal chromatic forest.
As we can color vertices in the chromatic forest more quickly than only using the \textsc{(3,2)-CSP} algorithm, this will improve the runtime of the algorithm.

\begin{definition}[Maximal High-Magnitude Chromatic Forest]
A \textit{maximal high-magnitude chromatic forest} is a maximal chromatic forest that covers all vertices in $U$, and all vertices in $N_3$ adjacent to a vertex in $U'$.
\end{definition}

\begin{lemma}
\label{lem:hmchromaticforest}
Let $G$ be a graph in which all vertex degrees are three or more, in which there is no cycle of degree-three vertices nor any connected subgraph of nine or more degree-three vertices. Let $F$ be a maximal low-magnitude bushy forest in $G$. Then, we can find a maximal high-magnitude chromatic forest in $G[V-F]$.
\end{lemma}
\begin{proof}
Beigel and Eppstein showed that we can find a chromatic forest in $G[V-F]$.
To find a chromatic forest, they first found a maximal forest of $K_{1,3}$ trees in $G[V-F]$: 
trees with a root vertex and three children.
While possible, they modify the forest of $K_{1,3}$ trees by removing one $K_{1,3}$ tree and adding at least two. 
Using this forest of $K_{1,3}$ trees, they showed that every vertex in $U$ can be assigned to one of the $K_{1,3}$ trees, such that no tree is assigned more than five grandchildren \cite[Lemma~25]{Beigel2005}.

We will now show that we can also include all high-magnitude vertices that are adjacent to a vertex in $U'$ in the chromatic forest. 
This will turn it from a maximal chromatic forest to a maximal high-magnitude chromatic forest.

All high-magnitude vertices or vertices in $U$ will have three neighbors outside the bushy forest. 
As such, if they are not included in the forest of $K_{1,3}$ trees, they are adjacent to at least one and at most three of these trees. 
They cannot be adjacent to no trees, as then we could add a new $K_{1,3}$ rooted at this vertex. 
Suppose a vertex $v$ is adjacent to $i$ $K_{1,3}$ trees, we give vertex $v$ weight $1/i$.
Beigel and Eppstein showed that we can find a maximal chromatic forest if the weight of all potential grandchildren does not exceed five for any tree.

Consider an arbitrary $K_{1,3}$ tree $T$ whose potential grandchildren have a weight over five.
For each of the three children, there can exist at most two grandchildren: at most six in total.
As each grandchild has weight at most one, there must be six potential grandchildren. Furthermore, at least five have unit weight and thus be adjacent to only one $K_{1,3}$ tree.

Suppose there exists a possible grandchild $v$ that neighbors a vertex that is not a possible grandchild of $T$.
Depending on the weight of $v$, we can always replace $T$ by two $K_{1,3}$ trees:
\begin{itemize}
    \item $v$ has a weight of one.
    We remove $T$ and create a new $K_{1,3}$ rooted at $v$.
    As $v$ neighbors at most one other grandchild of $T$, there must be a child vertex of $T$ whose grandchildren do not border $v$.
    We also create a new tree rooted at this child vertex.
    \item $v$ has a weight greater than one.
    Let $v'$ be the grandchild of $T$ that neighbors the same child of $T$ as $v$.
    We remove $T$ and create a new $K_{1,3}$ tree rooted at the parent of $v$.
    $v$ neighbors at most one other grandchild of $T$, while $v'$ neighbors at most two.
    $T$ had six potential grandchildren, so there must be one that does not border $v$ or $v'$:
    we can root a new $K_{1,3}$ tree at this vertex.
\end{itemize}

Let this neighbor of $v$ not be part of a $K_{1,3}$ tree. Then, we could remove $T$ and create two new trees: one rooted at $v$, and one at the child vertex of $T$ that does neighbor $v$.
Let $v$ instead border two $K_{1,3}$ trees.
In this case, we can also create two new trees.
The other grandchildren cannot be adjacent to another $K_{1,3}$ tree. So, we remove the original tree, and root one at the parent of $v$, and one at one of the grandchildren that are not adjacent to the parent of $v$. The only possible remaining configuration of vertices is shown in Figure \ref{fig:6ex2}

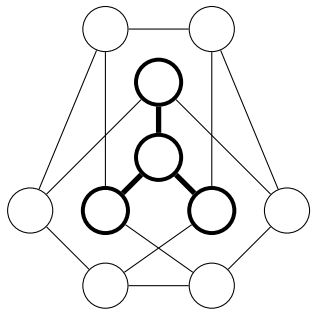
\begin{figure}[H]
\centering
\begin{tikzpicture}[main/.style = {draw, circle}, minimum size=0.6cm]
\node[main, line width=1.5pt] (1) {}; 
\node[main, line width=1.5pt] (2) [above of=1] {};
\node[main, line width=1.5pt] (3) [below left of=1] {};
\node[main, line width=1.5pt] (4) [below right of=1] {};
\node[main] (5) [above left of=2] {};
\node[main] (6) [above right of=2] {};
\node[main] (7) [left of=3] {};
\node[main] (8) [below of=3] {};
\node[main] (9) [right of=4] {};
\node[main] (10) [below of=4] {};

\draw[line width=2.0pt] (1) -- (2);
\draw[line width=2.0pt] (1) -- (3);
\draw[line width=2.0pt] (1) -- (4);
\draw[] (2) -- (7);
\draw[] (2) -- (9);
\draw[] (3) -- (5);
\draw[] (3) -- (10);
\draw[] (4) -- (6);
\draw[] (4) -- (8);
\draw[] (5) -- (6);
\draw[] (6) -- (9);
\draw[] (9) -- (10);
\draw[] (10) -- (8);
\draw[] (8) -- (7);
\draw[] (5) -- (7);
\end{tikzpicture}
\caption{The only possible configuration where a $K_{1,3}$ tree (bold) has grandchildren of weight larger than five within a maximal high-magnitude chromatic forest.}
\label{fig:6ex2}
\end{figure}

To find a $K_{1,3}$ tree whose potential grandchildren have weight over six, all six potential grandchildren must not have neighbors outside the $K_{1,3}$ tree.
All grandchildren must have three neighbors, one of which being a child vertex of $T$.
Each grandchild of $T$ must border two other grandchildren of $T$.
There must either be two cycles of length three through the grandchildren or one cycle of length six.
With two cycles, we could root two new trees, each rooted at one of the grandchildren.
So, there must exist a cycle of length six through the six grandchildren.

If all potential grandchildren are vertices in $U$, this would create a cycle of degree-three vertices. As such, in this situation, there must exist a high-magnitude vertex as a possible grandchild
We only consider high-magnitude vertices adjacent to vertices in $U'$ as potential grandchildren, so there must also exist at least one vertex in $U'$.
As a vertex in $U'$ is adjacent to three high-magnitude vertices, and neither the vertices in the $K_{1,3}$ tree nor the potential grandchildren can have neighbors outside this structure, there must be at least three high-magnitude vertices in the $K_{1,3}$ tree or its potential grandchildren.

Presume no further high-magnitude vertices or vertices in $U'$ are included.
Then, we can remove the current $K_{1,3}$ tree and create a new one rooted at the vertex in $U'$ instead, with the three high-magnitude vertices as its children.
Then, the tree will have fewer than six potential grandchildren, as only vertices in $U$ and high-magnitude vertices adjacent to vertices in $U'$ as potential grandchildren.
Otherwise, all grandchildren would be vertices in $U$.
In this case, the vertices in $U$ would form a cycle, which is not possible.

Consider the situation where the tree and its potential grandchildren contain more high-magnitude vertices adjacent to a vertex in $U'$.
Then, there must exist another vertex in $U'$. The high-magnitude vertices adjacent to this vertex in $U'$ cannot be adjacent to another tree in the bushy forest: 
there cannot exist two disjoint $K_{1,3}$ trees in the $K_{1,3}$ tree and its potential grandchildren.
Otherwise, we could replace the $K_{1,3}$ tree with two such trees. 
As such, the two vertices in $U'$ must share a common neighbor in $N_3$.

\begin{figure}[H]
\centering
\begin{tikzpicture}[main/.style = {draw, circle}, minimum size=0.6cm]
\node[main,fill=red] (1) {U}; 
\node[main,fill=green] (2) [above of=1] {N};
\node[main,fill=green] (3) [below left of=1] {N};
\node[main,fill=green] (4) [below right of=1] {N};
\node[main,fill=red] (5) [above left of=2] {U};
\node[main,fill=cyan] (6) [above right of=2] {N};
\node[main,fill=cyan] (7) [left of=3] {N};
\node[main,fill=red] (8) [below of=3] {U};
\node[main,fill=red] (9) [right of=4] {U};
\node[main,fill=cyan] (10) [below of=4] {N};
\node[main,fill=green] (11) [below left of=10] {L};
\node[main,fill=green] (12) [below right of=9] {L};
\node[main,fill=green] (13) [ left of=7] {L};
\node[main,fill=red] (14) [below of=11] {R};

\draw[] (1) -- (2);
\draw[] (1) -- (3);
\draw[] (1) -- (4);
\draw[] (2) -- (7);
\draw[] (2) -- (9);
\draw[] (3) -- (5);
\draw[] (3) -- (10);
\draw[] (4) -- (6);
\draw[] (4) -- (8);
\draw[] (5) -- (6);
\draw[] (6) -- (9);
\draw[] (9) -- (10);
\draw[] (10) -- (8);
\draw[] (8) -- (7);
\draw[] (5) -- (7);
\draw[] (11) -- (10);
\draw[] (11) -- (3);
\draw[] (12) -- (6);
\draw[] (12) -- (4);
\draw[] (13) -- (7);
\draw[] (13) -- (2);
\draw[] (14) -- (11);
\draw[] (14) -- (12);
\draw[] (14) -- (13);
\end{tikzpicture} 
\hspace{1cm}%
\begin{tikzpicture}[main/.style = {draw, circle}, minimum size=0.6cm]
\node[main,fill=red] (1) {U}; 
\node[main,fill=green] (2) [above of=1] {N};
\node[main,fill=cyan] (3) [below left of=1] {N};
\node[main,fill=green] (4) [below right of=1] {N};
\node[main,fill=red] (5) [above left of=2] {U};
\node[main,fill=cyan] (6) [above right of=2] {N};
\node[main,fill=cyan] (7) [left of=3] {N};
\node[main,fill=red] (8) [below of=3] {U};
\node[main,fill=red] (9) [right of=4] {U};
\node[main,fill=green] (10) [below of=4] {N};
\node[main,fill=cyan] (11) [below left of=10] {L};
\node[main,fill=green] (12) [below right of=9] {L};
\node[main,fill=green] (13) [ left of=7] {L};
\node[main,fill=red] (14) [below of=11] {R};

\draw[] (1) -- (2);
\draw[] (1) -- (3);
\draw[] (1) -- (4);
\draw[] (2) -- (7);
\draw[] (2) -- (9);
\draw[] (3) -- (5);
\draw[] (3) -- (10);
\draw[] (4) -- (6);
\draw[] (4) -- (8);
\draw[] (5) -- (6);
\draw[] (6) -- (9);
\draw[] (9) -- (10);
\draw[] (10) -- (8);
\draw[] (8) -- (7);
\draw[] (5) -- (7);
\draw[] (11) -- (10);
\draw[] (11) -- (3);
\draw[] (12) -- (6);
\draw[] (12) -- (4);
\draw[] (13) -- (7);
\draw[] (13) -- (2);
\draw[] (14) -- (11);
\draw[] (14) -- (12);
\draw[] (14) -- (13);
\end{tikzpicture} 
\caption{A valid 3-coloring for the only configuration possible where a tree in a maximal high-magnitude chromatic forest has a weight greater than five.}
\label{fig:6ex}
\end{figure}
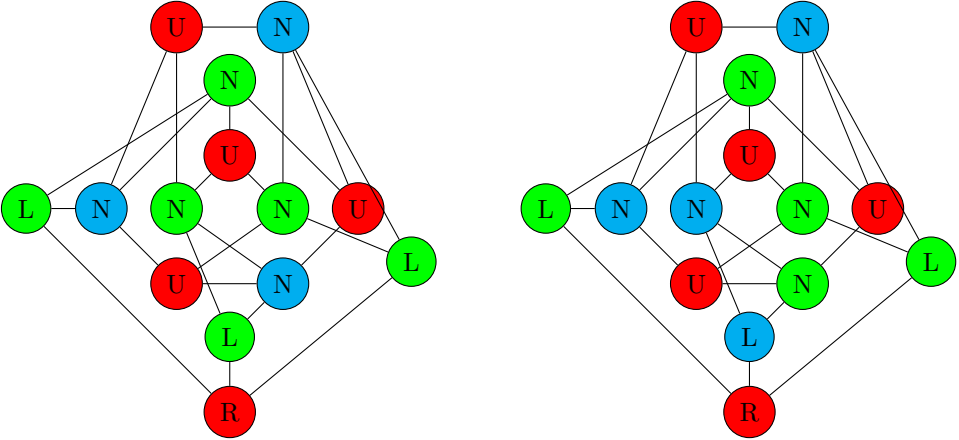

Now, the $K_{1,3}$ tree and its potential grandchildren must include multiple vertices in $U'$.
Then, the tree in the bushy forest they are adjacent to must neighbor at least six high-magnitude vertices.
The only possible configuration is shown in Figure \ref{fig:6ex}.

In this configuration, there must be three pairs of high-magnitude vertices that neighbor the same leaf of the tree in the bushy forest.
Then, the $K_{1,3}$ tree and its six grandchildren are only attached to the rest of the graph through these three leaves.
These three leaves share a neighbor in a single root vertex, so they must be 2-colorable.
Regardless of their coloration, it is always possible to color this configuration of vertices.
The three leaves either have the same color, or two have the same color and the third has a different color.
Figure \ref{fig:6ex} shows an exemplary coloring for both of these cases.

We now know that if the weight of the grandchildren of a $K_{1,3}$ tree is larger than five, we can either replace the tree and reduce the weight of the grandchildren, or the tree is trivial to color.
As such, we can conclude that we can find a maximal high-magnitude chromatic forest. 
The maximal high-magnitude chromatic forest allows us to color vertices in $N_3$ more easily for each vertex in $U'$.
Specifically, for every five vertices in $N_{3,5}$, there must be one vertex in $U'$ adjacent to three of the vertices in $N_{3,5}$.
For the vertices in $N_{3,j}$ for $6 \leq j \leq 8$, they will all be adjacent to a vertex in $U'$.
\end{proof}

We found constraints between the partitions of the vertices: these limit the number of vertices in $N$ and $U$, based on their properties and the number of vertices in $L$.
Furthermore, some high-magnitude vertices can always be colored by using the chromatic forest.
We now create a linear program to decide the worst-case graph for our algorithm.

\section{Analyzing our Algorithm}
\label{sec:lp}
We improve Beigel and Eppstein's algorithm by finding a maximal low-magnitude bushy forest and a maximal high-magnitude chromatic forest.
We enumerate all possible color assignments for selected vertices in these forests.

\lstset{
  breaklines=true,
  postbreak=\mbox{\textcolor{red}{$\hookrightarrow$}\space},
}
\begin{lstlisting}
Remove vertices of degree one and degree two.
Remove cycles of degree-three vertices and connected subgraphs of at least nine degree-three vertices.
Find a maximal low-magnitude bushy forest.
Find a maximal high-magnitude chromatic forest.
For each color assignment for the internal vertices of the bushy forest and selected vertices of the chromatic forest:
    Color all vertices according to the color assignment.
    Use the (3,2)-CSP algorithm on the remaining vertices.
\end{lstlisting}

We will now prove Theorem \ref{the:1}.
We know constraints on the partitions with different relations to the bushy forest from Section \ref{sec:analysis}.
We use these constraints, along with the time it takes to color each vertex of the graph, to formulate a linear program.
This linear program will find the worst-case graph for our algorithm and determine that it runs in time $\mathcal{O}^*(1.3217^n)$.

\subsection{Linear Program}
We found a set of constraints (Constraints \ref{align:1}, \ref{align:2}, \ref{align:3}, and \ref{align:4}) between the partitions of the vertices.
We use these constraints to create a linear program, which maximizes the runtime of the algorithm.

First, we will discuss the maximization function of the linear program and how it helps us find to discover the worst-case runtime of our algorithm.
Then, we discuss further constraints that are required for the linear program to function.
Finally, we will discuss how this linear program leads to our final result:
solving \textsc{3-coloring} in time $\mathcal{O}^*(1.3217^n)$.

\subsection{Creating the Linear Program}

\begin{align}
\max \quad
& \log(3)\cdot |R|+\log(2)\cdot |I|+\log(1.36443)\cdot |N^*|+\log(1.34004)\cdot|U^*| \label{ilp:obj}
\end{align}

\begin{align}
|N|-\frac{3}{5}\cdot|N_{3,5}|-|N_{3,6}|-|N_{3,7}|-|N_{3,8}| = |N^*| \label{ilp:7}\\ \quad
|U|+\frac{3}{5}\cdot|N_{3,5}|+|N_{3,6}| = |U^*| \label{ilp:8}
\end{align}

Recall that we iterate over all $3^{|R|}\cdot2^{|I|}$ color assignments for the internal vertices in the maximal low-magnitude bushy forest.
This allows us to remove the leaves ($|L|$) in polynomial time, so they do not increase the exponential runtime.
For the remaining vertices, all vertices in $U$ will be covered by the maximal high-magnitude chromatic forest, along with every high-magnitude vertex adjacent to a vertex in $U'$.
We define the set $U^*$ as all vertices that will be covered by the maximal high-magnitude chromatic forest and thus be colored in time $\mathcal{O}^*(1.34004^{|U^*|}$:
all vertices in $U$, along with all vertices in $N_3$ adjacent to vertices in $U'$ (Lemma \ref{lem:hmchromaticforest}).
The remaining vertices will be colored by the \textsc{(3,2)-CSP} algorithm.
Let $N^*$ be the set of vertices that are neither covered by the maximal low-magnitude bushy forest nor the maximal high-magnitude chromatic forest and will simply be colored using the \textsc{(3,2)-CSP} algorithm.
Indeed, $N^*$ equals all vertices in $N$, excluding high-magnitude vertices adjacent to vertices in $U$ (Lemma \ref{lem:hmchromaticforest}).

Notice that $3^{|R|}\cdot2^{|I|}\cdot1.36443^{|N^*|}\cdot1.34004^{|U^*|}$ calculates the runtime of the \textsc{3-coloring} algorithm, based on the partitions of the graph.
However, this is not a linear function.
Luckily, taking the logarithm of this function leads to a linear function: $\log(3^{|R|}\cdot2^{|I|}\cdot1.36443^{|N^*|}\cdot1.34004^{|U^*|})$.
Hence, this function determine the parameters $|R|$, $|I|$, $|N^*|$, and $|U^*|$ that maximize the runtime.

\begin{align}
|R|+|I|+|L|+|N|+|U|=n \label{ilp:1} \\ \quad
    \sum_{k=1}^3|N_k|=|N| \label{ilp:9}\\ \quad
\sum_{i=1}^8|N_{3,i}|=|N_3| \label{ilp:10}\\ \quad
\sum_{j=0}^7|U_j|+|U'|=|U| \label{ilp:11}\\ \quad
\sum_{i=1}^8(\frac{8}{i}\cdot|N_{3,i}|) \leq 8\cdot |R| \label{ilp:4}
\end{align}
Constraints \ref{align:1}, \ref{align:2}, \ref{align:3}, and \ref{align:4} function as constraints in the linear program. Constraint \ref{ilp:1} ensures that the partitions of the vertex set sum up to an arbitrary number of vertices $n$. 
We add further constraints (Constraints \ref{ilp:9}, \ref{ilp:10}, \ref{ilp:11}) that ensure that other partitions sum up to their parent set.
Finally, we add Constraint \ref{ilp:4}: 
by definition, each vertex in $N_{3,i}$ is adjacent to a tree in the maximal bushy forest that is adjacent to $i$ high-magnitude vertices in total.

\subsection{Results of the Linear Program}

Table \ref{tab:results} shows the results of the linear program.
In the results, sets of vertices are calculated as their fraction of the total number of vertices.
Notice that most sets of vertices are zero: they do not appear in the worst-case scenario. 
The vertices in $N$ (the neighbors of the maximal low-magnitude bushy forest) that do exist are in either $N_{3,6}$ or $N_{2}$.
None of the vertices in $N$ are in $N_1$, so all vertices in $U$ are either in $U_0$ or in $U'$. 
The vertices in $U_0$ exist in connected subgraphs of degree-three vertices of eight vertices in $U_0$: the largest connected subgraph of degree-three vertices that can exist.

\begin{table}[H]
\begin{tabular}{|ll|ll|ll|ll|}
\hline
$\mathbf{|R|}$  & 0.0396825 & $\mathbf{|E|}$   & 0.5555556 & $\mathbf{|N_{3,4}|}$ & 0         & $\mathbf{|U_2|}$ & 0         \\ \hline
$\mathbf{|I|}$  & 0         & $\mathbf{|N_1|}$  & 0         & $\mathbf{|N_{3,5}|}$ & 0         & $\mathbf{|U_3|}$ & 0         \\ \hline
$\mathbf{|L|}$  & 0.1587302 & $\mathbf{|N_2|}$  & 0.0396825 & $\mathbf{|N_{3,6}|}$ & 0.2380952 & $\mathbf{|U_4|}$ & 0         \\ \hline
$\mathbf{|N|}$  & 0.2777778 & $\mathbf{|N_3|}$  & 0.2380952 & $\mathbf{|N_{3,7}|}$ & 0         & $\mathbf{|U_5|}$ & 0         \\ \hline
$\mathbf{|U|}$  & 0.5238095 & $\mathbf{|N_{3,1}|}$ & 0         & $\mathbf{|N_{3,8}|}$ & 0         & $\mathbf{|U_6|}$ & 0         \\ \hline
$\mathbf{|N^*|}$ & 0.0396825 & $\mathbf{|N_{3,2}|}$ & 0         & $\mathbf{|U_0|}$  & 0.4444444 & $\mathbf{|U_7|}$ & 0         \\ \hline
$\mathbf{|U^*|}$ & 0.7619048 & $\mathbf{|N_{3,3}|}$ & 0         & $\mathbf{|U_1|}$  & 0         & $\mathbf{|U'|}$ & 0.0793651 \\ \hline
\end{tabular}
\caption{Results of the linear program.}
\label{tab:results}
\end{table}

The worst-case graph consists of many trees, which are all the same. 
Each tree is surrounded by six high-magnitude vertices. 
The remaining neighbors of the tree are vertices in $N_2$.
Every tree will cause two vertices in $U'$ to exist, with all other vertices not adjacent to a tree being in $U_0$.
One such tree is displayed in Figure \ref{fig:worstcase}.

Specifically, for every vertex in $R$, we have the following number of vertices in other sets:
\begin{enumerate}
    \item $L$: 4
    \item $N_2$: 1
    \item $N_{3,6}$: 6
    \item $U_0$: 11.2
    \item $U'$: 2
\end{enumerate}

By using the frequency of these sets of vertices, we determine that the worst-case runtime of our algorithm is $\mathcal{O}^*((3^1\cdot1.36443^1\cdot1.34004^{19.2})^{n/25.2})=\mathcal{O}^*(1.3217^n)$.

\section{Conclusion}
\label{sec:conclusion}
In this paper, we presented an improved algorithm for \textsc{3-coloring}. 
The new algorithm performs the following steps:

Notably, we introduced the concept of the high-magnitude vertex. 
In Beigel and Eppstein's algorithm, the existence of high-magnitude vertices allowed for many vertices to be outside the maximal bushy forest. 
Instead, we created the maximal low-magnitude bushy forest: a maximal bushy forest with few high-magnitude vertices.
Additionally, we also created the maximal high-magnitude chromatic forest.
Using these forests, we determine a set of vertices for which we iterate over all possible color assignments.
All vertices not included in this set, nor neighboring this set, will be solved by the \textsc{(3,2)-CSP} algorithm.

We performed a sophisticated analysis of the different properties of the vertices in the graph, defined by their relation to the bushy forest.
This analysis, formulated as a linear program, allowed us to generate the graph for which the runtime of the algorithm would be the largest: time $\mathcal{O}^*(1.3217^n)$.
This analysis helped us find the following structure, a single tree in the family of worst-case graphs:


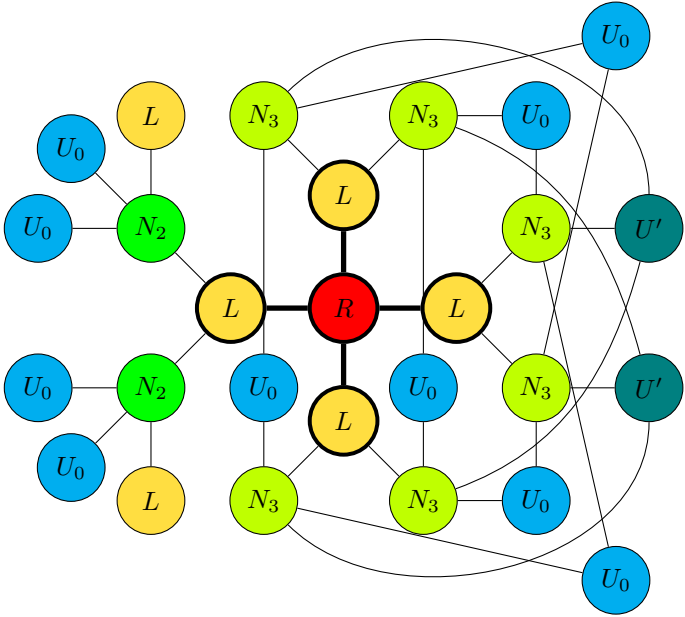
\begin{figure}[H]
\centering
\begin{tikzpicture}[main/.style = {draw, circle},node distance=1.5cm, minimum size=0.9cm]
\node[main] (1) [fill=red, line width=1.5pt] {$R$}; 
\node[main] (2) [above of=1,fill=Goldenrod, line width=1.5pt] {$L$};
\node[main] (3) [right of=1,fill=,fill=Goldenrod, line width=1.5pt] {$L$}; 
\node[main] (4) [below of=1,fill=Goldenrod, line width=1.5pt] {$L$};
\node[main] (5) [left of=1,fill=Goldenrod, line width=1.5pt] {$L$}; 
\node[main] (6) [above left of=2,fill=lime] {$N_3$}; 
\node[main] (7) [above right of=2,fill=lime] {$N_3$}; 
\node[main] (8) [above right of=3,fill=lime] {$N_3$}; 
\node[main] (9) [below right of=3,fill=lime] {$N_3$}; 
\node[main] (10) [below right of=4,fill=lime] {$N_3$}; 
\node[main] (11) [below left of=4,fill=lime] {$N_3$}; 
\node[main] (12) [below left of=5,fill=green] {$N_2$};
\node[main] (13) [above left of=5,fill=green] {$N_2$};
\node[main] (14) [right of=7,fill=cyan] {$U_0$};
\node[main] (15) [above right of=14,fill=cyan] {$U_0$};
\node[main] (16) [right of=10,fill=cyan] {$U_0$};
\node[main] (17) [below right of=16,fill=cyan] {$U_0$};
\node[main] (18) [right of=12,fill=cyan] {$U_0$};
\node[main] (19) [left of=9,fill=cyan] {$U_0$};
\node[main] (20) [below left of=12,fill=cyan] {$U_0$};
\node[main] (21) [left of=12,fill=cyan] {$U_0$};
\node[main] (22) [left of=13,fill=cyan] {$U_0$};
\node[main] (23) [above left of=13,fill=cyan] {$U_0$};
\node[main] (24) [right of=8,fill=teal] {$U'$};
\node[main] (25) [right of=9,fill=teal] {$U'$};
\node[main] (26) [below of=12,fill=Goldenrod] {$L$};
\node[main] (27) [above of=13,fill=Goldenrod] {$L$};

\draw[line width=2.0pt] (1) -- (2);
\draw[line width=2.0pt] (1) -- (3);
\draw[line width=2.0pt] (1) -- (4);
\draw[line width=2.0pt] (1) -- (5);
\draw (2) -- (6);
\draw (2) -- (7);
\draw (3) -- (8);
\draw (3) -- (9);
\draw (4) -- (10);
\draw (4) -- (11);
\draw (5) -- (12);
\draw (5) -- (13);

\draw (6) -- (15);
\draw (6) to [out=45,in=90,looseness=1] (24);
\draw (6) -- (18);
\draw (7) -- (14);
\draw (7) to [out=-20, in=105] (25);
\draw (7) -- (19);
\draw (8) -- (14);
\draw (8) -- (24);
\draw (8) -- (17);
\draw (9) -- (15);
\draw (9) -- (25);
\draw (9) -- (16);
\draw (10) -- (16);
\draw (10) to [out=20, in=-105] (24);
\draw (10) -- (19);
\draw (11) -- (17);
\draw (11) to [out=-45,in=-90,looseness=1] (25);
\draw (11) -- (18);
\draw (12) -- (20);
\draw (12) -- (21);
\draw (12) -- (26);
\draw (13) -- (22);
\draw (13) -- (23);
\draw (13) -- (27);

\end{tikzpicture} 
\caption{Worst-case scenario for \textsc{3-coloring} (single tree). \textcolor{red}{Red}: root vertices ($R$). \textcolor{orange}{Orange}: internal vertices ($I$). \textcolor{Goldenrod}{Yellow}: leaves ($L$). \textcolor{lime}{Lime}: high-magnitude vertices. \textcolor{green}{Green}: other neighbors to the bushy forest ($N$). \textcolor{teal}{Teal}: vertices outside the bushy forest with only high-magnitude vertices as neighbors ($U'$). \textcolor{cyan}{Cyan}: other vertices ($U$).}
\label{fig:worstcase}
\end{figure}

We used the \textsc{(3,2)-Constraint Satisfaction Problem} algorithm as a black box algorithm, so any improvement to the \textsc{(3,2)-Constraint Satisfaction Problem} will automatically improve our algorithm for \textsc{3-coloring}.
Our algorithm will color many vertices through the bushy and chromatic forests, so the input of the \textsc{(3,2)-Constraint Satisfaction Problem} will be relatively predictable.
One might be able to improve the algorithm for \textsc{3-coloring}, by investigating whether the input to the \textsc{(3,2)-Constraint Satisfaction Problem} contains easily structures that are easy to reduce.

It is unclear whether the concepts presented in this paper could be applied to improve the runtime of $k$-\textsc{coloring} for $k \geq 4$.
Especially for \textsc{4-coloring}, the best known algorithm for \textsc{(4,2)-Constraint Satisfaction Problem} is the exact same as the \textsc{(3,2)-Constraint Satisfaction Problem}.
However, the \textsc{(4,2)-Constraint Satisfaction Problem} adds a complex new case:
both variables with three and four possible colors require exponential time, but the variables with four colors are more difficult to color than those with three.
If a vertex has a colored neighbor, it will still have three possible colors.
So, we cannot reduce a vertex in polynomial time, unless it has two neighbors which received different colors.
As such, it is significantly more complex to use a partially colored graph and the \textsc{(4,2)-Constraint Satisfaction Problem} to solve \textsc{4-coloring}.

\section{Acknowledgements}
\label{sec:8}
We thank Carla Groenland and Jesper Nederlof for their detailed feedback on earlier drafts of this paper and helpful discussions about graph coloring.
We thank Till Miltzow and Ivan Bliznets for giving recommendations about the writing style of this paper.

\printbibliography
\end{document}